\documentclass[a4paper,twoside,11pt]{article}
\usepackage[margin=2.9cm]{geometry}
\usepackage[affil-it]{authblk}

\usepackage{amsmath,amssymb,amsthm,amsfonts}
\usepackage{graphicx}
\usepackage{url,xspace,soul,wrapfig}
\usepackage{subfig,caption}
\usepackage{paralist,enumerate,enumitem}
\usepackage{todonotes}

\graphicspath{{figures/}}

\newcommand{\sref}[1]{\protect\subref{#1}}

\newcommand{\naesat}{\textsc{naesat}\xspace}

\newcommand{\TWOcolorability}{2-colorability\xspace}
\newcommand{\THREEcoloring}{\textsc{3-coloring}\xspace}
\newcommand{\THREEcolorings}{3-colorings\xspace}
\newcommand{\THREEcolorable}{3-colorable\xspace}

\newcommand{\coloring}[2]{tree-diameter-$#2$ $#1$-coloring\xspace}
\newcommand{\colorings}[2]{tree-diameter-$#2$ $#1$-colorings\xspace}
\newcommand{\colorable}[2]{tree-diameter-$#2$ $#1$-colorable\xspace}

\newcommand{\konecoloring}[1]{(edge, $#1$)-coloring\xspace}

\newcommand{\konecolorable}[1]{(edge, $#1$)-colorable\xspace}

\newcommand{\kcoloring}[1]{(star, $#1$)-coloring\xspace}
\newcommand{\kcolorings}[1]{(star, $#1$)-colorings\xspace}
\newcommand{\kcolorable}[1]{(star, $#1$)-colorable\xspace}
\newcommand{\twocoloring}{\kcoloring{2}}
\newcommand{\twocolorings}{\kcolorings{2}}
\newcommand{\twocolorable}{\kcolorable{2}}
\newcommand{\threecoloring}{\kcoloring{3}}

\newcommand{\threecolorable}{\kcolorable{3}}

\title{Vertex-Coloring with Star-Defects}
\author[1]{Patrizio~Angelini}
\author[1]{Michael~A.~Bekos}
\author[1]{Michael~Kaufmann}
\author[2]{Vincenzo~Roselli}
\affil[1]{Wilhelm-Schickhard-Institut f\"ur Informatik, Universit\"at T\"ubingen, Germany\\
$\{$angelini,bekos,mk$\}$@informatik.uni-tuebingen.de}
\affil[2]{Dipartimento di Ingegneria, Universit{\`a} Roma Tre,
Italy\\~~~~~~~~roselli@dia.uniroma3.it}
\date{}
\newtheorem{lemma}{Lemma}
\newtheorem{theorem}{Theorem}

\begin{document}
\maketitle

\begin{abstract}
\emph{Defective coloring} is a variant of traditional
vertex-coloring, according to which adjacent vertices are allowed to
have the same color, as long as the monochromatic components induced
by the corresponding edges have a certain structure. Due to its 
important applications, as for example in the bipartisation of graphs, this
type of coloring has been extensively studied, mainly with respect to
the size, degree, and acyclicity of the monochromatic components.

In this paper we focus on defective colorings in which the
monochromatic components are acyclic and have small diameter,
namely, they form stars. For outerplanar graphs, we give a
linear-time algorithm to decide if such a defective coloring exists
with two colors and, in the positive case, to construct one. Also, we prove that
an outerpath (i.e., an outerplanar graph whose weak-dual is a path)
always admits such a two-coloring. Finally, we present
NP-completeness results for non-planar and planar graphs of bounded
degree for the cases of two and three colors.
\end{abstract}

\section{Introduction}
\label{sec:introduction}

Graph coloring is a fundamental problem in graph theory, which has
been extensively studied over the years (see,
e.g.,~\cite{DBLP:books/daglib/0077283} for an overview). Most of the
research in this area has been devoted to the \emph{vertex-coloring
problem} (or \emph{coloring problem}, for short), which dates back to
1852~\cite{MM12}. In its general form, the problem asks to label the
vertices of a graph with a given number of colors, so that no two
adjacent vertices share the same color. In other words, a coloring of
a graph partitions its vertices into a particular number of
independent sets (each of these sets is usually referred to as a
\emph{color class}, as all its vertices have the same color). A
central result in this area is the so-called \emph{four color
theorem}, according to which every planar graph admits a coloring
with at most four colors; see e.g.~\cite{MR2463991}. Note that the
problem of deciding whether a planar graph is $3$-colorable is
NP-complete~\cite{DBLP:books/fm/GareyJ79}, even for graphs of maximum
degree~$4$~\cite{d-uccp-80}.

Several variants of the coloring problem have been proposed over the
years. One of the most studied is the so-called \emph{defective
coloring}, which was independently introduced by Andrews and
Jacobson~\cite{aj-gcn-85}, Harary and Jones~\cite{hj-cc-85}, and
Cowen et al.~\cite{cgj-97}. In the defective coloring problem edges
between vertices of the same color class are allowed, as long as the
monochromatic components induced by vertices of the same color
maintain some special structure. In this respect, one can regard the
classical vertex-coloring as a defective one in which every
monochromatic component is an isolated vertex, given that every
color class defines an independent set. In this work we focus on
defective colorings in which each component is acyclic and has
small diameters. In particular, we call a graph $G$
\emph{\colorable{\kappa}{\lambda}} if the vertices of $G$ can be
colored with $\kappa$ colors, so that all monochromatic components
are acyclic and of diameter at most~$\lambda$, where $\kappa \geq
1$, $\lambda \geq 0$. Clearly, a classical $\kappa$-coloring
corresponds to a \coloring{\kappa}{0}.
The \emph{diameter of a coloring} is defined as the maximum diameter
among the monochromatic components.

We present algorithmic and complexity results for
\colorings{\kappa}{\lambda} for small values of $\kappa$ and
$\lambda = 2$. For simplicity, we refer to this problem as
\emph{\kcoloring{\kappa}}, as each monochromatic component is a
\emph{star} (i.e., a tree with diameter two; see
Figure~\ref{fig:star}). Similarly, we refer to the
\coloring{\kappa}{\lambda} problem when $\lambda=1$ as
\emph{\konecoloring{\kappa}} problem. By definition, a
\konecoloring{\kappa} is also a \kcoloring{\kappa}.
Figs.\ref{fig:example1}-\ref{fig:example3} show a trade-off between
number of colors and structure of the monochromatic components.

Our work can be seen as a variant of the \emph{bipartisation} of
graphs, namely the problem of making a graph bipartite by removing a
small number of elements (e.g, vertices or edges), which is a central
graph problem with many applications~\cite{Hadlock75,Karp72}. The
bipartisation by removal (a not-necessarily minimal number of)
\emph{non-adjacent} edges corresponds to the \konecoloring{2}
problem. In the \kcoloring{2} problem, we also solve some kind of
bipartisation by removing independent stars. Note that we do not ask
for the minimum number of removed stars but for the existence of a
solution.

\begin{figure}[tb]
	\centering
	\subfloat[\label{fig:example1}]{\includegraphics[scale=1.4,page=1]{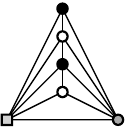}}
	\hfil
	\subfloat[\label{fig:example2}]{\includegraphics[scale=1.4,page=2]{figures}}
	\hfil
	\subfloat[\label{fig:example3}]{\includegraphics[scale=1.4,page=3]{figures}}
	\hfil
	\subfloat[\label{fig:star}]{\includegraphics[scale=1.4,page=4]{figures}}
	\caption{(a-c) Different colorings of the same graph:
	(a)~a traditional $4$-coloring, 
	(b)~an \konecoloring{3}
	(c)~a \twocoloring; 
	(d)~a star with three leaves; its \emph{center} has degree~$3$.}
	\label{fig:examples}
\end{figure}

To the best of our knowledge, this is the first time that the
defective coloring problem is studied under the requirement of
having color classes of small diameter. Previous research was
focused either on their size or their
degree~\cite{aj-gcn-85,cgj-97,hj-cc-85,Lo66,Lovasz1975269}. As
byproducts of these works, one can obtain several results for the
\konecoloring{\kappa} problem. More precisely, from a result of
Lov\'{a}sz~\cite{Lo66}, it follows that all graphs of maximum degree
$4$ or $5$ are \konecolorable{3}. However, determining whether a
graph of maximum degree~$7$ is \konecolorable{3} is
NP-complete~\cite{cgj-97}. In the same work, Cowen et
al.~\cite{cgj-97} prove that not all planar graphs are
\konecolorable{3} and that the corresponding decision problem is
NP-complete, even in the case of planar graphs of maximum
degree~$10$. Results for graphs embedded on general surfaces are
also known~\cite{a-ndc-87,ccw-86,cgj-97}. Closely related is also
the so-called \emph{tree-partition-width} problem, which is a
variant of the defective coloring problem in which the
graphs induced by each color class must be
acyclic~\cite{Ding199645,ehno-aepofbtg-11,Wood20091245}, i.e., there
is no restriction on their diameter. Our contributions are:
\begin{itemize}
\item In Section~\ref{sec:algorithms}, we present a linear-time
algorithm to determine whether an outerplanar graph is \twocolorable.
Note that outerplanar graphs are \THREEcolorable \cite{ps-evecog-86},
and hence \threecolorable, but not necessarily \twocolorable. On the
other hand, we can always construct \twocolorings for outerpaths
(which form a special subclass of outerplanar graphs whose
weak-dual\footnote{Recall that the \emph{weak-dual} of a plane graph
is the subgraph of its dual induced by neglecting the face-vertex
corresponding to its unbounded face.} is a path).

\item In Section~\ref{sec:np-completeness}, we prove that the
\twocoloring problem is NP-complete, even for graphs of maximum
degree $5$ (note that the corresponding \konecoloring{2} problem is
NP-complete, even for graphs of maximum degree $4$~\cite{cgj-97}).
Since all graphs of maximum degree $3$ are
\konecolorable{2}~\cite{Lo66}, this result leaves open only the case
for graphs of maximum degree~$4$. We also prove that the
\threecoloring problem is NP-complete, even for graphs of maximum
degree $9$ (recall that the corresponding \konecoloring{3} problem is
NP-complete, even for graphs of maximum degree~$7$~\cite{cgj-97}).
Since all graphs of maximum degree $4$ or $5$ are
\konecolorable{3}~\cite{Lo66}, our result implies that the
computational complexity of the \threecoloring problem remains
unknown only for graphs of maximum degree $6$, $7$, and $8$. For
planar graphs, we prove that the \twocoloring problem remains
NP-complete even for triangle-free planar graphs (recall that
triangle-free planar graphs are always \THREEcolorable~\cite{Kow10},
while the test of \TWOcolorability can be done in linear~time).
\end{itemize}

\section{Coloring Outerplanar Graphs and Subclasses}
\label{sec:algorithms}
In this section we consider \twocolorings of outerplanar graphs. To
demonstrate the difficulty of the problem, we first give an example
(see Figure~\ref{lem:outercounterexample}) of a small outerplanar
graph not admitting any \twocoloring. Therefore, in
Theorem~\ref{thm:outerplanar} we study the complexity of deciding
whether a given outerplanar graph admits such a coloring and present
a linear-time algorithm for this problem; note that outerplanar
graphs always admit \THREEcolorings~\cite{ps-evecog-86}. Finally, we
show that a notable subclass of outerplanar graphs, namely
outerpaths, always admit \twocolorings by providing a constructive
linear-time algorithm (see Theorem~\ref{thm:outerpaths}).

\begin{lemma}\label{lem:outercounterexample}
There exist outerplanar graphs that are not \twocolorable.
\end{lemma}
\begin{proof}
We prove that the outerplanar graph of
Figure~\ref{fig:outercouterexample} is not \twocolorable. In
particular, we show that in any $2$-coloring of this graph there
exists a monochromatic path of four vertices.
Assume w.l.o.g.~that vertex $u$ has color gray. Then, at least two
vertices out of $u_1,\dots,u_8$ are gray, as otherwise there would be
a path of four white vertices. Hence, $u$ is the center of a gray star.

Next, we observe that either $u_2$ is white or the path
$u_{21},\dots,u_{24}$ must consist of only white vertices.
Similarly, we observe that either $u_3$ is white or the path
$u_{31},\dots,u_{34}$ must consist of only white vertices.
If both $u_2$ and $u_3$ are white, then either one of paths
$u_{21},\dots,u_{24}$ and $u_{31},\dots,u_{34}$ consists only of gray
vertices, or there exists a path from one of $u_{21},\dots,u_{24}$
via $u_2$ and $u_3$ to one of $u_{31},\dots,u_{34}$, that consists
only of white vertices. Clearly, all aforementioned cases lead to a
monochromatic path of four vertices.
\end{proof}

\begin{figure}[htb]
	\centering
	\subfloat[\label{fig:outercouterexample}]{\includegraphics[page=5]{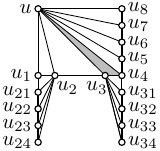}}
	\hfil
	\subfloat[\label{fig:outer-defs}]{\includegraphics[page=6]{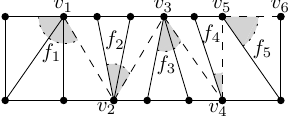}} 
	\caption{(a)~An outerplanar graph that is not \twocolorable. 
	(b)~An outerpath, whose spine edges are drawn as dashed segments.
	Dotted arcs highlighted in gray correspond to edges belonging to
	the fan of each spine vertex. Note that $|f_6|=0$.}
	\label{fig:outerpath}
\end{figure}

Lemma~\ref{lem:outercounterexample} implies that not all
outerplanar graphs are \twocolorable. In the following we give a
linear-time algorithm to decide whether an outerplanar graph is
\twocolorable and in case of an affirmative answer to compute the
actual coloring.

\begin{theorem}\label{thm:outerplanar}
Given an outerplanar graph $G$, there exists a linear-time
algorithm to test whether $G$ admits a \twocoloring and to construct
a \twocoloring, if one exists.
\end{theorem}
\begin{proof}
We assume that $G$ is embedded according to its outerplanar
embedding. We can also assume that $G$ is biconnected. This is not a
loss of generality, as we can always reduce the number of
cut-vertices by connecting two neighbors $a$ and $b$ of a cut-vertex
$c$ belonging to two different biconnected components with a path
having two internal vertices. Clearly, if the augmented graph is
\twocolorable, then the original one is \twocolorable. For the other
direction, given a \twocoloring of the original graph, we can obtain
a corresponding coloring of the augmented graph by coloring the
neighbors of $a$ and $b$ with different color than the ones of $a$
and $b$, respectively.

Denote by $T$ the weak dual of $G$ and root it at a leaf $\rho$ of
$T$. For a node $\mu$ of $T$, we denote by $G(\mu)$ the subgraph of
$G$ corresponding to the subtree of $T$ rooted at $\mu$. We also
denote by $f(\mu)$ the face of $G$ corresponding to $\mu$ in $T$. If
$\mu \neq \rho$, consider the parent $\nu$ of $\mu$ in $T$ and their
corresponding faces $f(\nu)$ and $f(\mu)$ of $G$, and let $(u,v)$ be
the edge of $G$ shared by $f(\nu)$ and $f(\mu)$. We say that $(u,v)$
is the \emph{attachment edge} of $G(\mu)$ to $G(\nu)$.
The attachment edge of the root $\rho$ is any edge of face $f(\rho)$
that is incident to the outer face (since $G$ is biconnected and
$\rho$ is a leaf, this edge always exists). Consider a \twocoloring
of $G(\mu)$. In this coloring, each of the endpoints $u$ and $v$ of
the attachment edge of $G(\mu)$ plays exactly one of the following
roles:
\begin{inparaenum}[$(i)$] 
\item \emph{center} or 
\item \emph{leaf} of a colored star; 
\item \emph{isolated vertex}, that is, it has no neighbor with the
same color; or 
\item \emph{undefined}, that is, the only neighbor of $u$ (resp.
$v$) which has its same color is $v$ (resp. $u$). Note that if the only
neighbor of $u$ (resp. $v$) which has its same color is different
from $v$ (resp. from $u$), we consider $u$ (resp. $v$) as a center.
\end{inparaenum}
Two \twocolorings of $G(\mu)$ are \emph{equivalent} w.r.t. the
attachment edge $(u,v)$ of $G(\mu)$ if in the two \twocolorings each
of $u$ and $v$ has the same color and plays the same role.
This definition of equivalence determines a partition of the
colorings of $G(\mu)$ into a set of equivalence classes. Since both
the number of colors and the number of possible roles of each vertex
$u$ and $v$ are constant, the number of different equivalence
classes is also constant 
(note that, when the role is undefined, $u$ and $v$ must have the
same color).

In order to test whether $G$ admits a \twocoloring, we perform a
bottom-up traversal of $T$.
When visiting a node $\mu$ of $T$ we compute the maximal set
$C(\mu)$ of equivalence classes such that, for each class $C \in
C(\mu)$, graph $G(\mu)$ admits at least a coloring belonging to $C$.
Note that $|C(\mu)| \le 38$.
In order to compute $C(\mu)$, we consider the possible equivalence
classes one at a time, and check whether $G(\mu)$ admits a
\twocoloring in this class, based on the sets
$C(\mu_1),\dots,C_(\mu_h)$ of the children $\mu_1,\dots,\mu_h$ of
$\mu$ in $T$, which have been previously computed.
In particular, for an equivalence class $C$ we test the existence of
a \twocoloring of $G(\mu)$ belonging to $C$ by selecting an
equivalence class $C_i \in C(\mu_i)$ for each $i = 1,\dots,h$ in
such a way that:
\begin{enumerate}
\item the color and the role of $u$ in $C_1$ are the same as the
ones $u$ has in $C$;
\item the color and the role of $v$ in $C_h$ are the same as the
ones $v$ has in $C$;
\item for any two consecutive children $\mu_i$ and $\mu_{i+1}$, let
$x$ be the vertex shared by $G(\mu_i)$ and $G(\mu_{i+1})$. Then, $x$
has the same color in $C_i$ and $C_{i+1}$ and, if $x$ is a leaf in
$C_i$, then $x$ is isolated in $C_{i+1}$ (or vice-versa); and
\item for any three consecutive children $\mu_{i-1}$, $\mu_i$, and
$\mu_{i+1}$, let $x$ (resp. $y$) be the vertex shared by
$G(\mu_{i-1})$ and $G(\mu_i)$ (resp. by $G(\mu_{i})$ and
$G(\mu_{i+1})$). Then, $x$ (resp. $y$) has the same color in $C_i$
and $C_{i-1}$ (resp. $C_{i+1}$); also, if $x$ and $y$ are both
undefined in $C_i$, then in $C_{i-1}$ and $C_{i+1}$ none of $x$ and
$y$ is a leaf, and at least one of them is isolated.
\end{enumerate}

Note that the first two conditions ensure that the coloring belongs
to $C$, while the other two ensure that it is a \twocoloring. Since
the cardinality of $C(\mu_i)$ is bounded by a constant, the test
can be done in linear time. If the test succeeds, add $C$ to
$C(\mu)$. 

Once all $38$ equivalence classes are tested, if $C(\mu)$ is empty,
then we conclude that $G$ is not \twocolorable. Otherwise we proceed
with the traversal of $T$. At the end of the traversal, if $C(\rho)$
is not empty, we conclude that $G$ is \twocolorable. A \twocoloring
of $G$ can be easily constructed by traversing $T$ top-down, by
following the choices performed during the bottom-up visit.
\end{proof}
In the following, we consider a subclass of outerplanar graphs,
namely outerpaths, and we prove that they always admit \twocolorings.
Note that the example that we presented in
Lemma~\ref{lem:outercounterexample} is ``almost'' an outerpath,
meaning that the weak-dual of this graph contains only degree-$1$ and
degree-$2$ vertices, except for one specific vertex that has 
degree~$3$ (see the face of Figure~\ref{fig:outercouterexample}
highlighted in gray). Recall that the weak-dual of an outerpath is a
path (hence, it consists of only degree-$1$ and degree-$2$ vertices).

Let $G$ be an outerpath (see Figure~\ref{fig:outer-defs}). We assume
that $G$ is inner-triangulated. This is not a loss of generality, as
any \twocoloring of a triangulated outerpath induces a \twocoloring
of any of its subgraphs. We first give some definitions. We call
\emph{spine vertices} the vertices $v_1,v_2,\dots,v_m$ that have
degree at least four in $G$. We consider an additional spine vertex
$v_{m+1}$, which is the (unique) neighbor of $v_m$ along the cycle
delimiting the outer face that is not adjacent to $v_{m-1}$. Note
that the spine vertices of $G$ induce a path, that we call
\emph{spine} of $G$\footnote{Note that the spine of $G$ coincides
with the spine of the caterpillar obtained from the outerpath $G$ by
removing all the edges incident to its outer face, neglecting the
additional spine vertex $v_{m+1}$.}. The \emph{fan} $f_i$ of a spine
vertex $v_i$ consists of the set of neighbors of $v_i$ in $G$, except
for $v_{i-1}$ and for those following and preceding $v_i$ along the
cycle delimiting the outer face\footnote{Fan $f_i$ contains all the
leaves of the caterpillar incident to $v_i$, plus the following spine
vertex $v_{i+1}$.}; note that $|f_i|\ge 1$ for each $i=1,\dots,m$,
while $|f_{m+1}|=0$. For each $i = 1, \dots,m+1$, we denote by $G_i$
the subgraph of $G$ induced by the spine vertices $v_1,\dots,v_i$ and
by the fans $f_1,\dots,f_{i-1}$. Note that $G_{m+1}=G$. We denote by
$c_i$ the color assigned to spine vertex $v_i$, and by $c(G_i)$ a
coloring of graph $G_i$. Finally, we say that an edge of $G$ is
\emph{colored} if its two endpoints have the same color.

\begin{figure}[tb]
	\centering
	\includegraphics[page=7]{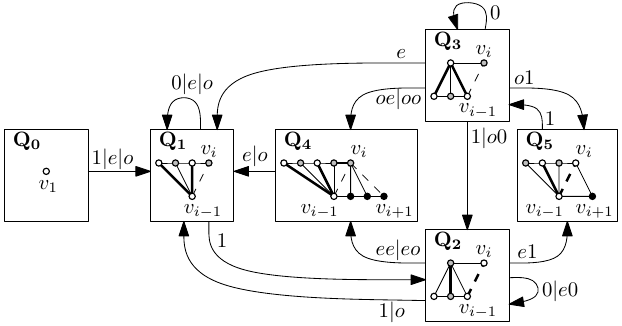}
	\caption{Schematization of the algorithm. Each node represents the
	(unique) condition satisfied by $G_i$ at some step $0\le i\le k$.
	An edge label $0,1,e,o$ represents the fact that the cardinality of
	a fan $f_i$ is $0$, $1$, even $\neq 0$, or odd $\neq 1$. If the
	label contains two characters, the second one describes the cardinality of $f_{i+1}$.
	An edge between $Q_j$ and $Q_h$ with label $x\in\{1,e,o\}$ (with
	label $xy$, where $y\in\{0,1,e,o\}$) represents the fact that, if
	$G_i$ satisfies condition $Q_j$ and $|f_i|=x$ (resp. $|f_i|=x$ and
	$|f_{i+1}|=y$), then $f_i$ is colored so that $G_{i+1}$ satisfies
	$Q_h$.}\label{fig:automaton}
\end{figure}

\begin{theorem}\label{thm:outerpaths}
Every outerpath admits a \twocoloring, which can be computed in
linear time.
\end{theorem}
\begin{proof}
Let $G$ be an outerpath with spine $v_1,\dots,v_k$. We describe an
algorithm to compute a \twocoloring of $G$. At each step $i = 1,
\dots, k$ of the algorithm we consider the spine edge
$(v_{i-1},v_{i})$, assuming that a \twocoloring of $G_i$ has already
been computed satisfying one of the following conditions (see
Figure~\ref{fig:automaton}):
\begin{description}
\item[$Q_0$:] The only colored vertex is $v_1$;
\item[$Q_1$:] $c_{i}\neq c_{i-1}$, vertex $v_{i-1}$ is the center of a star with color $c_{i-1}$, and no colored edge is incident to $v_{i}$;
\item[$Q_2$:] $c_{i}=c_{i-1}$, and no colored edge other than $(v_{i-1}, v_i)$ is incident to $v_{i-1}$ or $v_{i}$;
\item[$Q_3$:] $c_{i}\neq c_{i-1}$, vertex $v_{i-1}$ is a leaf of a star with color $c_{i-1}$, and no colored edge is incident to $v_{i}$;
\item[$Q_4$:] $c_{i}\neq c_{i-1}$, vertex $v_{i-1}$ is the center of a star with color $c_{i-1}$, and vertex $v_{i}$ is the center of a star with color $c_{i}$; further, $i<k$ and $|f_i|>1$;
\item[$Q_5$:] $c_{i}=c_{i-1}$, vertex $v_{i-1}$ is the center of a star with color $c_{i-1}$, and no colored edge other than $(v_{i-1}, v_i)$ is incident to $v_{i}$; further, $i<k$ and $|f_i|=1$.
\end{description}

Next, we color the vertices in $f_i$ in such a way that $c(G_{i+1})$
is a \twocoloring satisfying one of the conditions; refer to
Figure~\ref{fig:automaton} for a schematization of the case
analysis.
In the first step of the algorithm, we assign an arbitrary color to
$v_1$, and hence $c(G_1)$ satisfies $Q_0$. For $i=1,\dots,k$ we
color $f_i$ depending on the condition satisfied by $c(G_i)$.

\smallskip

\noindent{\textbf{Coloring }$\mathbf{c(G_i)}$ \textbf{satisfies}
$\mathbf{Q_0}$}: Independently of the cardinality of $f_i$, we color
its vertices with alternating colors so that $c_{i+1} \neq c_i$. In
this way, the only possible colored edges are incident to $v_i$ and
not to $v_{i+1}$. So, $c(G_{i+1})$ satisfies condition $Q_1$.

\smallskip

\noindent{\textbf{Coloring} $\mathbf{c(G_i)}$ \textbf{satisfies}
$\mathbf{Q_1}$}: In this case we distinguish the following subcases,
based on the cardinality of $f_i$.

\begin{itemize}

\item If $|f_i|=0$, we have that $i=k$ and hence $G_k=G$. It follows
that $c(G_k)$ is a \twocoloring of $G$.

\item If $|f_i|=1$ (that is, $f_i$ contains only $v_{i+1}$; see
Figure~\ref{fig:q1_1_post}), we set $c_{i+1}=c_i$. Since the only
neighbor of $v_{i+1}$ in $G_{i+1}$ different from $v_i$ is $v_{i-1}$, whose
color is $c_{i-1} \neq c_i$, and since $v_i$ has no neighbor with
color $c_i$ other than $v_{i+1}$, by condition $Q_1$, coloring
$c(G_{i+1})$ is a \twocoloring satisfying condition $Q_2$.

\begin{figure}[tb]
	\centering
	\subfloat[\label{fig:q1_1_post}]{\includegraphics[page=8]{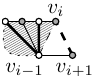}}
	\hfil
	\subfloat[\label{fig:q1_gt1_post}]{\includegraphics[page=9]{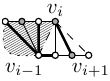}}
	\hfil
	\subfloat[\label{fig:q2_odd_post}]{\includegraphics[page=10]{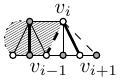}}
	\hfil
	\subfloat[\label{fig:q2_even_zero}]{\includegraphics[page=11]{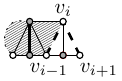}}
	\hfil
	\subfloat[\label{fig:q2_even_one_post}]{\includegraphics[page=12]{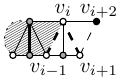}}
	\hfil
	\subfloat[\label{fig:q2_even_gt_one_post}]{\includegraphics[page=13]{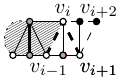}}
	\caption{Graph $G_{i+1}$ after coloring $f_i$ when $c(G_i)$ satisfies: 
		$Q_1$ and 
		\sref{fig:q1_1_post} $|f_i|=1$ or 
		\sref{fig:q1_gt1_post} $|f_i|>1$;
		$Q_2$ and 
		\sref{fig:q2_odd_post} $|f_i|=o$, 
		or $|f_i|=e$ and 
		\sref{fig:q2_even_zero} $|f_{i+1}|=0$, 
		\sref{fig:q2_even_one_post} $|f_{i+1}|=1$, or
		\sref{fig:q2_even_gt_one_post}  $f_{i+1} > 1$.
		Shaded regions represent $G_i$. Bold edges connect vertices with the same color, while spine edges are dashed.
	}
\end{figure}

\item If $|f_i|>1$ (see Figure~\ref{fig:q1_gt1_post}), we color the
vertices in $f_i$ with alternating colors so that $c_{i+1}\neq c_i$.
This implies that every colored edge of $G_{i+1}$ not belonging to
$G_i$ is incident either to $v_i$, if its color is $c_i$, or to
$v_{i-1}$, if its color is $c_{i-1}$; the latter case only happens
if $|f_i|$ is odd. Thus, $v_i$ (resp. $v_{i-1}$) is the center of a
star of color $c_i$ (resp. $c_{i-1}$) in $G_{i+1}$.
Since $v_i$ has no neighbor with color $c_i$ in $G_i$, while
$v_{i-1}$ is a center also in $G_i$, coloring $c(G_{i+1})$ is a
\twocoloring. Finally, since $v_{i+1}$ has no neighbors with color
$c_{i+1} \neq c_i$, by construction, $c(G_{i+1})$ satisfies
condition $Q_1$.
\end{itemize}
\noindent{\textbf{Coloring }$\mathbf{c(G_i)}$ \textbf{satisfies
}$\mathbf{Q_2}$:} We again distinguish subcases based on $|f_i|$.

\begin{itemize}
\item If $|f_i|=0$, we have that $i=k$ and hence $c(G_k)$ is a
\twocoloring of $G=G_k$.

\item If $|f_i|$ is odd, including the case $|f_i|=1$ (see
Figure~\ref{fig:q2_odd_post}), we color the vertices of $f_i$ with
alternating colors in such a way that $c_{i+1}\neq c_i$. By
construction, $c(G_{i+1})$ is a \twocoloring satisfying condition
$Q_1$.

\item If $|f_i|$ is even and different from $0$, instead, we have to
consider the cardinality of $f_{i+1}$ in order to decide the
coloring of $f_i$. We distinguish three subcases:

\smallskip

\begin{description} 
\item[$|f_{i+1}|=0$]: Note that in this case $i=k$ holds (see
Figure~\ref{fig:q2_even_zero}). We color the vertices of $f_i$ with
alternating colors so that $c_{i+1} = c_i$. Note that the unique
neighbor of $v_{i-1}$ in $f_i$ has color different from $c_{i-1}$,
since $|f_i|$ is even. Hence, all the new colored edges are incident
to $v_i$, which implies that $c(G_{i+1})$ is a \twocoloring
satisfying condition $Q_2$.
	
\item[$|f_{i+1}|=1$]: Note that $i<k$ and $f_{i+1}$ only contains
$v_{i+2}$ (see Figure~\ref{fig:q2_even_one_post}). We color the
vertices of $f_i$ with alternating colors so that $c_{i+1} = c_i$. Since (i)
all the new colored edges are incident to $v_i$, (ii) $v_i$ and
$v_{i-1}$ have no neighbor with their same color in $G_i$ (apart
from each other), (iii) $c_{i+1} = c_i$, and (iv) $i<k$, we have
that $c(G_{i+1})$ is a \twocoloring satisfying condition $Q_5$.

\item[$|f_{i+1}|>1$]: Note that $i<k$ (see
Figure~\ref{fig:q2_even_gt_one_post}). Independently of whether
$|f_{i+1}|$ is even or odd, we color the vertices of $f_i$ so that
$c_{i+1} \neq c_i$, the unique neighbor of $v_{i+1}$ different from
$v_i$ has also color $c_{i+1}$, and all the other vertices have
alternating colors. Since each new colored edge is incident to
either $v_i$ or $v_{i+1}$, since $c_{i+1} \neq c_i$, and since
$i<k$, coloring $c(G_{i+1})$ is a \twocoloring satisfying condition
$Q_4$.
\end{description}
\end{itemize}

\smallskip

\begin{figure}[tb]
	\centering
	\subfloat[\label{fig:q3_one}]{\includegraphics[page=14]{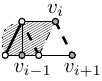}}
	\hfil
	\subfloat[\label{fig:q3_even}]{\includegraphics[page=15]{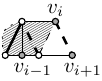}}
	\hfil
	\subfloat[\label{fig:q4}]{\includegraphics[page=16]{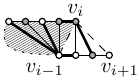}}
	\hfil
	\subfloat[\label{fig:q5}]{\includegraphics[page=17]{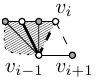}}
	\caption{Graph $G_{i+1}$ after coloring $f_i$ when $c(G_i)$ satisfies: 
	$Q_3$ and
		\sref{fig:q3_one} $|f_i|=1$, or
		\sref{fig:q3_even} $|f_i|=e$;
	$Q_4$ \sref{fig:q4}; or
	$Q_5$ \sref{fig:q5}.
	Shaded regions represent $G_i$. Bold edges connect vertices with the same color, while spine edges are dashed.
	}
\end{figure}

\noindent{\textbf{Coloring} $\mathbf{c(G_i)}$ \textbf{satisfies}
$\mathbf{Q_3}$}:
\begin{itemize}
\item If $|f_i|=0$, we have that $i=k$ and hence $c(G_k)$ is a
\twocoloring of $G=G_k$.

\item If $|f_i|=1$ (that is, $f_i$ contains only $v_{i+1}$; see
Figure~\ref{fig:q3_one}), we set $c_{i+1}=c_i$.
As in the analogous case in which $c(G_i)$ satisfies condition
$Q_1$, we can prove that $c(G_{i+1})$ is a \twocoloring which
satisfies condition $Q_2$.

\item If $|f_i|$ is even and different from $0$ (see
Figure~\ref{fig:q3_even}), we color the vertices of $f_i$ with alternating
colors in such a way that $c_{i+1}\neq c_i$. By construction,
$c(G_{i+1})$ is a \twocoloring which satisfies condition $Q_1$.

\item If $|f_i|$ is odd and different from $1$, we again
consider the cardinality of $f_{i+1}$ in order to decide the
coloring of $f_i$. For the four possible classes of values
of $|f_{i+1}|$, the coloring strategy and the condition satisfied by
the resulting coloring $c(G_{i+1})$ are the same as for
the analogous case in which $c(G_i)$ satisfies $Q_2$ and $|f_i|$ is
even.
\end{itemize}

\noindent{\textbf{Coloring} $\mathbf{c(G_i)}$ \textbf{satisfies}
$\mathbf{Q_4}$}:
Note that $|f_i|>0$, given that $i<k$, and $|f_i|\neq 1$, by
condition $Q_4$. Independently of whether $|f_i|$ is even or odd
(see Figure~\ref{fig:q4}), we color the vertices in $f_i$ with alternating
colors so that $c_{i+1} \neq c_i$. In this way, the only possible
colored edges are incident to $v_{i-1}$ and to $v_i$, which are both
centers of a star already in $G_i$, and not to $v_{i+1}$. Hence,
$c(G_{i+1})$ is a \twocoloring satisfying condition $Q_1$.

\smallskip

\noindent{\textbf{Coloring} $\mathbf{c(G_i)}$ \textbf{satisfies}
$\mathbf{Q_5}$}:
Note that $|f_i| = 1$, by condition $Q_5$ (that is, $f_i$ only
contains $v_{i+1}$; see Figure~\ref{fig:q5}). We set $c_{i+1} \neq c_i$;
clearly, $c(G_{i+1})$ is a \twocoloring satisfying condition $Q_3$.

\smallskip

Observe that the running time of the algorithm is linear in the
number of vertices of $G$. In fact, at each step $i = 1,\dots,k$,
the condition $Q_j$ satisfied by $c(G_i)$ and the cardinalities of
$f_i$ and $f_{i+1}$ are known (the cardinality of all the fans can
be precomputed in advance), and the coloring strategy to obtain
$c(G_{i+1})$ and the condition satisfied by this coloring are
uniquely determined by these information in constant time.
\end{proof}

\section{NP-completeness for (Planar) Graphs of Bounded Degree}
\label{sec:np-completeness}

In this section, we study the computational complexity of the
\twocoloring and \threecoloring problems for (planar) graphs of
bounded degree.

\begin{theorem} 
It is NP-complete to determine whether a graph admits a
\twocoloring, even in the case where its maximum degree is no more
than $5$.
\label{thm:2colorDeg5NpHard}
\end{theorem}  
\begin{proof}
The problem clearly belongs to NP; a non-deterministic algorithm
only needs to guess a color for each vertex of the graph and then in
linear time can trivially check whether the graphs induced by each
color-set are forests of stars. To prove that the problem is
NP-hard, we employ a reduction from the well-known Not-All-Equal $3$-SAT
problem or \naesat for short~\cite[p.187]{Pap07}. An instance of
\naesat consists of a $3$-CNF formula $\phi$ with variables
$x_1,\ldots,x_n$ and clauses $C_1,\ldots,C_m$. The task is to find a
truth assignment of $\phi$ so that no clause has all three literals
\emph{equal} in truth value (that is, not all are true). We show how
to construct a graph $G_\phi$ of maximum vertex-degree $5$ admitting
a \twocoloring if and only if $\phi$ is satisfiable. Intuitively,
graph $G_\phi$ reflecting formula $\phi$ consists of a set of
subgraphs serving as variable gadgets that are connected to simple
$3$-cycles that serve as clause gadgets in an appropriate way; see
Figure~\ref{fig:reduction1} for an example.

Consider the graph of Figure~\ref{fig:k24}, which contains two adjacent
vertices, denoted by $u_1$ and $u_2$, and four vertices, denoted by
$v_1$, $v_2$, $v_3$ and $v_4$, that form a path, so that each of
$u_1$ and $u_2$ is connected to each of $v_1$, $v_2$, $v_3$ and
$v_4$.
We claim that in any \twocoloring of this graph $u_1$ and $u_2$ have
different colors. Assume to the contrary that $u_1$ and $u_2$ have
the same color, say white. Since $u_1$ and $u_2$ are adjacent, none
of $v_1$, $v_2$, $v_3$ and $v_4$ is white. So, $v_1, \ldots, v_4$
form a monochromatic component in gray which is of diameter~$3$; a
contradiction. Hence, $u_1$ and $u_2$ have different colors, say
gray and white, respectively.
In addition, the colors of $v_1$, $v_2$, $v_3$ and $v_4$ alternate
along the path $v_1 \rightarrow v_2 \rightarrow v_3 \rightarrow
v_4$, as otherwise there would exist two consecutive vertices $v_i$
and $v_{i+1}$, with $i=1,2,3$, of the same color, which would create a monochromatic
triangle with either $u_1$ or $u_2$.

\begin{figure}[tb]
	\begin{minipage}[b]{.5\textwidth}
		\centering 
		\begin{minipage}[b]{\textwidth}
			\centering 
			\subfloat[\label{fig:k24}{variable-gadget}]{\includegraphics[page=18]{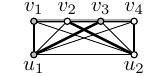}}
		\end{minipage}
		\begin{minipage}[b]{\textwidth}
			\centering 
			\subfloat[\label{fig:chain}{a chain of length $3$}]{\includegraphics[width=\textwidth,page=19]{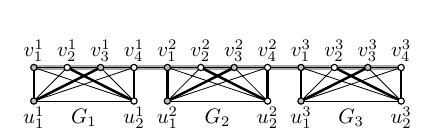}}
		\end{minipage}
	\end{minipage}
	\hfil
	\begin{minipage}[b]{.4\textwidth}
		\centering
		\subfloat[\label{fig:reduction1}{reduction; clause-gadgets are gray}]{\includegraphics[width=\textwidth,page=20]{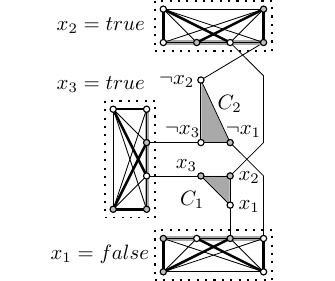}}
	\end{minipage}
    \caption{Illustration of:
    \sref{fig:k24}~a graph with $6$ vertices, 
    \sref{fig:chain}~a chain of length $3$, 
    \sref{fig:reduction1}~the reduction from \naesat to \twocoloring:
    $\phi = (x_1 \lor x_2 \lor x_3) \land (\neg x_1 \lor \neg x_2 \lor \neg x_3)$.
	The solution corresponds to the assignment $x_1=false$ and $x_2=x_3=false$.
	Sets $O_{x_1}$, $E_{x_2}$ and $E_{x_3}$ ($E_{x_1}$, $O_{x_2}$
	and $O_{x_3}$, resp.) are colored gray (white, resp.).}
    \label{fig:2colorDeg5NpHard}
\end{figure}

For $k \geq 1$, we form a \emph{chain of length $k$} that contains
$k$ copies $G_1, G_2, \ldots, G_k$ of the graph of
Figure~\ref{fig:k24}, connected to each other as follows (see
Figure~\ref{fig:chain}). For $i=1,2,\ldots,k$, let $u_1^i$, $u_2^i$,
$v_1^i$, $v_2^i$, $v_3^i$ and $v_4^i$ be the vertices of $G_i$.
Then, for $i=1,2,\ldots,k-1$ we introduce between $G_i$ and
$G_{i+1}$ an edge connecting vertices $v_4^i$ and $v_1^{i+1}$
(dotted in Figure~\ref{fig:chain}). This edge ensures that $v_4^i$
and $v_1^{i+1}$ are not of the same color, since otherwise we would have a
monochromatic path of length four. Hence, the colors of the
vertices of the so-called \emph{spine-path} $v_1^1 \rightarrow v_2^1
\rightarrow v_3^1 \rightarrow v_4^1 \rightarrow \ldots \rightarrow
v_1^k \rightarrow v_2^k \rightarrow v_3^k \rightarrow v_4^k$
alternate along this path. In other words, if the odd-positioned
vertices of the spine-path are white, then the even-positioned ones
will be gray, and vice versa. In addition, all vertices of the
spine-path have degree~$4$ (except for $v_1^1$ and $v_4^k$, which
have degree~$3$).

For each variable $x_i$ of $\phi$, graph $G_\phi$ contains a
so-called \emph{variable-chain} $\mathcal{C}_{x_i}$ of length
$\lceil \frac{n_i-2}{2} \rceil$, where $n_i$ is the number of
occurrences of $x_i$ in $\phi$, $1 \leq i \leq n$; see
Figure~\ref{fig:reduction1}. Let $O[\mathcal{C}_{x_i}]$ and
$E[\mathcal{C}_{x_i}]$ be the sets of odd- and even-positioned
vertices along the spine-path of $\mathcal{C}_{x_i}$, respectively.
For each clause $C_i=(\lambda_j \lor \lambda_k \lor \lambda_\ell)$
of~$\phi$, $1\leq i \leq m$, where $\lambda_j \in \{x_j,\neg x_j\}$,
$\lambda_k \in \{x_k,\neg x_k\}$, $\lambda_\ell \in \{x_\ell,\neg
x_\ell\}$ and $j, k, \ell \in \{1,\ldots,n\}$, graph $G_\phi$ 
contains a $3$-cycle of corresponding \emph{clause-vertices} which,
of course, cannot have the same color (\emph{clause-gadget};
highlighted in gray in Figure~\ref{fig:reduction1}). If $\lambda_j$ is
positive (negative), then we connect the clause-vertex corresponding
to $\lambda_j$ in $G_\phi$ to a vertex of degree less than $5$ that
belongs to set $E[\mathcal{C}_{x_j}]$ ($O[\mathcal{C}_{x_j}]$) of
chain $\mathcal{C}_{x_j}$. Similarly, we create connections for
literals $\lambda_k$ and $\lambda_\ell$; see the edges leaving the
triplets for clauses $C_1$ and $C_2$ in Figure~\ref{fig:reduction1}.
The length of $\mathcal{C}_{x i}$, $1 \leq i \leq n$, guarantees
that all connections are accomplished so that no vertex of
$\mathcal{C}_{x_i}$ has degree larger than $5$. Thus, $G_\phi$ is of
maximum degree~$5$. Since $G_\phi$ is linear in the size of $\phi$,
the construction can be done in $O(n + m)$ time. 

We show that $G_\phi$ is \twocolorable if and only if $\phi$ is
satisfiable. Assume first that $\phi$ is satisfiable.
If $x_i$ is true (false), then we color $E[\mathcal{C}_{x_i}]$ white
(gray) and $O[\mathcal{C}_{x_i}]$ gray (white). Hence,
$E[\mathcal{C}_{x_i}]$ and $O[\mathcal{C}_{x_i}]$ admit different
colors, as desired. Further, if $x_i$ is true (false), then we color
gray (white) all the clause-vertices of $G_\phi$ that correspond to
positive literals of $x_i$ in $\phi$ and we color white (gray) those
corresponding to negative literals. Thus, a clause-vertex of
$G_\phi$  cannot have the same color as its neighbor at the
variable-gadget. Since in the truth assignment of $\phi$ no clause
has all three literals true, no three clause-vertices belonging to
the same clause have the same color.

Suppose that $G_\phi$ is \twocolorable. By construction, each of
$E[\mathcal{C}_{x_i}]$ and $O[\mathcal{C}_{x_i}]$ is either white or
gray, $i=1,\ldots,n$. If $P[\mathcal{C}_{x_i}]$ is white, then we
set $x_i=true$; otherwise, we set $x_i=false$. Assume, to the
contrary, that there is a clause of $\phi$ whose literals are all
true or all false. By construction, the corresponding
clause-vertices of $G_\phi$, which form a $3$-cycle in $G_\phi$,
have the same color, which is a contradiction.
\end{proof}
We now turn our attention to planar graphs. Our proof follows the
same construction as the one of Theorem~\ref{thm:2colorDeg5NpHard}
but to ensure planarity we replace the crossings with appropriate
crossing-gadgets. Also, recall that the construction in
Theorem~\ref{thm:2colorDeg5NpHard} highly depends on the presence of
triangles (refer, e.g., to the clause gadgets). In the following
theorem, we prove that the \twocoloring problem remains NP-complete,
even in the case of triangle-free planar graphs. Note that in order
to avoid triangular faces, we use slightly more complicated variable-
and clause-gadgets, which have higher degree but still bounded by a
constant.

\newcommand{\colorTriangFreeNpHard}{
It is NP-complete to determine whether a triangle-free planar graph
admits a \twocoloring.}

\begin{theorem} 
\colorTriangFreeNpHard
\label{thm:2colorTriangFreeNpHard}
\end{theorem} 
\begin{proof}
Membership in NP can be shown as in the proof of
Theorem~\ref{thm:2colorDeg5NpHard}. To prove that the problem in
NP-hard, we again employ a reduction from \naesat. To avoid
crossings we will construct a triangle-free planar graph $G_\phi$
(with different variable- and clause-gadgets) similar to the
previous construction, so that $G_\phi$ admits a \twocoloring if and
only if $\phi$ is satisfiable.

The clause-gadget is illustrated in Fig.\ref{fig:clause-gadget}. It
consists of a $2 \times 3$ grid (highlighted in gray) and one vertex
of degree~$2$ (denoted by~$u$ in Fig.\ref{fig:clause-gadget})
connected to the top-left and bottom-right vertices of the grid. We
claim that the \emph{clause-vertices} of this gadget (denoted by
$u$, $u_{11}$ and $u_{23}$ in Fig.\ref{fig:clause-gadget}) cannot
all have the same color. For a proof by contradiction assume that
$u$, $u_{11}$ and $u_{23}$ are all black. Since $u_{12}$, $u_{13}$,
$u_{21}$ and $u_{22}$ are adjacent either to $u_{11}$ or to
$u_{23}$, none of them is black. Hence, $u_{21} \rightarrow u_{22}
\rightarrow u_{12} \rightarrow u_{13}$ is a monochromatic path of
length three; a contradiction to the diameter of the coloring.

\begin{figure}[tb]
    \centering 
    \subfloat[\label{fig:clause-gadget}]{\includegraphics[page=24]{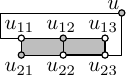}}
    \hfil
    \subfloat[\label{transmitter-gadget}]{\includegraphics[page=25]{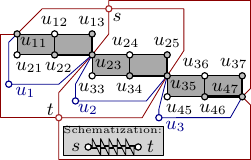}}
    \hfil
    \subfloat[\label{fig:variable-gadget}]{\includegraphics[page=26]{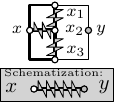}}
    \hfil
    \subfloat[\label{fig:chain-gadget}]{\includegraphics[page=27]{figures}}
    \hfil
    \subfloat[\label{fig:crossing-gadget}]{\includegraphics[page=28]{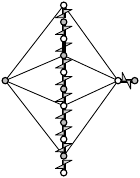}}
    \caption{ 
    (a)~clause-gadget,
    (b)~transmitter-gadget,
    (c)~variable-gadget,
    (d)~a chain of length~11,
    (e)~crossing-gadget.} 
    \label{fig:3colorTrianFreeNpHard}
\end{figure}

Fig.\ref{transmitter-gadget} illustrates the so-called
\emph{transmitter-gadget}, which consists of three copies of the $2
\times 3$ grid (highlighted in gray), each of which gives rise to a
clause-gadget with the degree-$2$ vertices $u_1$, $u_2$ and $u_3$.
It also has two additional vertices (denoted by $s$ and $t$ in
Fig.\ref{transmitter-gadget}), each of which forms a clause-gadget
with each of the three copies of the rectangular grid. We claim that
in any \twocoloring of the transmitter-gadget $s$ and $t$ are of the
same colors. Otherwise, a simple observation shows that there is a
monochromatic path of length three; a contradiction to the diameter
of the coloring. A schematization of the transmitter-gadget is given
in Fig.\ref{transmitter-gadget}.

The variable-gadget is illustrated in Fig.\ref{fig:variable-gadget}.
We claim that in any \twocoloring of these gadget vertices $x$ and
$y$ must be of different colors. Assume to the contrary that $x$ and
$y$ are both white. Then, vertices $x_1$, $x_2$ and $x_3$ must also
be white, due to the transmitter-gadgets involved. Hence, $x
\rightarrow x_1 \rightarrow y \rightarrow x_3 \rightarrow x_1$ is a
white-colored cycle; a contradiction to the diameter of the
coloring. A schematization of the variable-gadget is given in
Fig.\ref{fig:variable-gadget}. The corresponding one for the chain
is given in Fig.\ref{fig:chain-gadget}.

Since we proved that the clause-vertices of the clause-gadgets
cannot all have the same color and that the variable gadget has two
specific vertices of different colors, the rest of the construction
is identical to the one of the previous theorem. Note, however, that
$G_\phi$ is unlikely to be planar, as required by this theorem.
However we can arrange the variable-gadgets and the clause-gadgets
so that the only edges that cross are the ones joining the
variable-gadgets with the clause-gadgets. Then, we replace every
crossing by the crossing-gadget illustrated in
Fig.\ref{fig:crossing-gadget}. This particular gadget has the
following two properties:
\begin{inparaenum}[(i)]
\item its topmost and bottommost vertices must be of the same color
(due to the vertical arrangement of the variable-gadgets), which
implies that
\item the leftmost and rightmost vertices must be of the same color
as well.
\end{inparaenum}
Hence, we can replace all potential crossings with the
crossing-gadget. Since the number of crossings is quadratic to
the number of edges, the size of the construction is still
polynomial. Everything else in the construction and in the argument
remains the same. 
\end{proof}
Note that Theorems~\ref{thm:2colorDeg5NpHard} and
\ref{thm:2colorTriangFreeNpHard} have been independently proven by
Dorbec et al.~\cite{Dorbec14}. In the following theorem we prove
that the \twocoloring problem remains NP-complete even if one allows one
more color and the input graph is either of maximum degree~$9$ or
planar of maximum degree~$16$. Recall that all planar graphs are
$4$-colorable.

\begin{theorem} 
It is NP-complete to determine whether a graph $G$ admits a
\threecoloring, even in the case where the maximum degree of $G$ is
no more than $9$ or in the case where $G$ is a planar graph of maximum
degree $16$.
\label{thm:3colorDeg9NpHard}
\end{theorem} 
\begin{proof}
Membership in NP can be proved similarly to the corresponding one of
Theorem~\ref{thm:2colorDeg5NpHard}. To prove that the problem is
NP-hard, we employ a reduction from the well-known \THREEcoloring
problem, which is NP-complete even for planar graphs of maximum
vertex-degree $4$~\cite{d-uccp-80}. So, let $G$ be an instance of
the \THREEcoloring problem. To prove the first part of the theorem,
we will construct a graph $H$ of maximum vertex-degree $9$
admitting a \threecoloring if and only if $G$ is \THREEcolorable.

Central in our construction is the complete graph on six vertices
$K_6$, which is \threecolorable; see Figure~\ref{fig:k6}. We claim that
in any \threecoloring of $K_6$ each vertex is adjacent to exactly
one  vertex of the same color. For a proof by contradiction, assume
that there is a \threecoloring of $K_6$ in which three vertices, say
$u$, $v$ and $w$, have the same color.
From the completeness of $K_6$, it follows that $u$, $v$ and $w$
form a monochromatic components of diameter~$3$, which is a
contradiction.

\begin{figure}[tb]
    \centering
   \begin{minipage}[b]{.32\textwidth}
   	\centering
   	\subfloat[\label{fig:k6}{}]{\includegraphics[width=.8\textwidth,page=21]{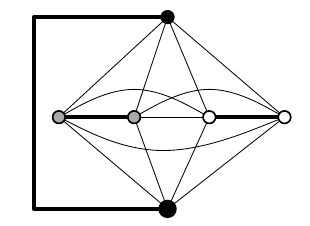}}
   \end{minipage}
   \hfil
   \begin{minipage}[b]{.32\textwidth}
   	\centering
   	\subfloat[\label{fig:3colorattachment}{}]{\includegraphics[width=.8\textwidth,page=22]{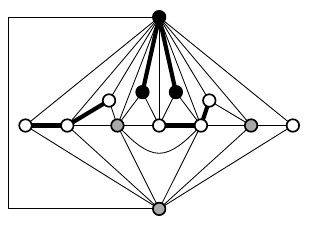}}
   \end{minipage}
   \hfil
   \begin{minipage}[b]{.32\textwidth}
   	\centering
   	\subfloat[\label{fig:counterexample}{}]{\includegraphics[width=.8\textwidth,page=23]{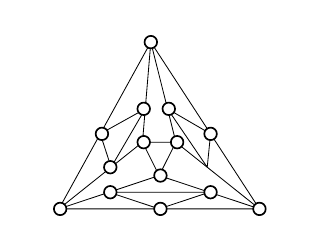}}
   \end{minipage}
   \caption{
    (a)~The complete graph on six vertices $K_6$. 
    (b)~The attachment-graph for the planar case.
    (c)~A planar graph of max-degree~$4$ that is not \twocolorable.}
    \label{fig:3colorNpHard}
\end{figure}

Graph $H$ is obtained from $G$ by attaching a copy of $K_6$ at each
vertex $u$ of $G$, and by identifying $u$ with a vertex of $K_6$,
which we call \emph{attachment-vertex}. Hence, $H$ has maximum degree~$9$.
As $H$ is linear in the size of $G$, it can be constructed in linear
time.

If $G$ admits a $3$-coloring, then $H$ admits a \threecoloring in
which each attachment-vertex in $H$ has the same color as the
corresponding vertex of $G$, and the colors of the other vertices
are determined based on the color of the attachment-vertices.
To prove that a \threecoloring of $H$ determines a \THREEcoloring of
$G$, it is enough to prove that any two adjacent attachment-vertices
$v$ and $w$ in $H$ have different colors, which clearly holds since
both $v$ and $w$ are incident to a vertex of the same color in the
corresponding copies of $K_6$ associated with them.

For the second part of the theorem, we attach at each vertex of $G$
the planar graph of Figure~\ref{fig:3colorattachment} using as
attachment its topmost vertex, which is of degree~$12$ (instead of
$K_6$ which is not planar). Hence, the constructed graph $H$ is
planar and has degree~$16$ as desired. Furthermore, it is not
difficult to be proved that in any \threecoloring of the graph of
Figure~\ref{fig:3colorattachment} its attachment-vertex is always
incident to (at least one) another vertex of the same color, that
is, it has exactly the same property with any vertex of $K_6$.
Hence, the rest of the proof is analogous to the one of the
first part of the theorem.
\end{proof} 

\section{Conclusions}
\label{sec:conclusions}

In this work, we presented algorithmic and complexity results for
the \twocoloring and the \threecoloring problems. We proved that all
outerpaths are \twocolorable and we gave a polynomial-time algorithm
to determine whether an outerplanar graph is \twocolorable. For the
classes of graphs of bounded degree and planar triangle-free graphs
we presented several NP-completeness results. However, there exist
several open questions raised by our work.
\begin{itemize}
\item In Theorem~\ref{thm:2colorDeg5NpHard} we proved that it is
NP-complete to determine whether a graph of maximum degree~$5$ is
\twocolorable. So, a reasonable question to ask is whether one can
determine in polynomial time whether a graph of maximum degree~$4$
is \twocolorable. The question is of relevance even for planar
graphs of maximum degree~$4$. Note that not all planar graphs of
maximum degree~$4$ are \twocolorable (Figure~\ref{fig:counterexample}
shows such a counterexample found by extensive case analysis), while
all graphs of maximum degree~$3$ are \konecolorable{2}~\cite{Lo66}.

\item Other classes of graphs, besides the outerpaths, that are
always \twocolorable are of interest. 

\item In Theorem~\ref{thm:3colorDeg9NpHard} we proved that it is
NP-complete to determine whether a graph of maximum degree~$9$ is
\threecolorable. The corresponding question on the complexity
remains open for the classes of graphs of maximum degree~$6$,
$7$ and $8$. Recall that graphs of maximum degree~$4$ or $5$ are
always \threecolorable. 

\item One possible way to expand the class of graphs that admit
defective colorings, is to allow larger values on the diameter of
the graphs induced by the same color class.
\end{itemize}

\paragraph{Acknowledgement:} We thank the participants of the
special session GNV of IISA'15 inspiring this work. We also thank
Pascal Ochem who brought~\cite{Dorbec14} to our attention, where some
of our NP-completeness results have been independently proven.

\bibliographystyle{abbrv}
\bibliography{starcoloring}

\end{document}